\documentclass{llncs}
\usepackage{graphicx}
\usepackage{makeidx}

\usepackage{color,amsfonts,graphics,latexsym,amsmath,amssymb,graphicx,%subfigure,
alltt}
{\begin{list}{}%
         {\setlength{\leftmargin}{#1}}%
         \item[]%
}
{\end{list}}

\usepackage{epsfig}
\usepackage{tabularx}
\usepackage{latexsym}

\def\qed{$\Box$}
\usepackage{enumerate}

\begin{document}

\title{NP-Hardness of Speed Scaling with a Sleep State}
\author{Gunjan Kumar, Saswata Shannigrahi}
\institute{Indian Institute of Technology Guwahati, India. \email{\{k.gunjan,saswata.sh\}@iitg.ernet.in} }
\maketitle

%.................Start of Abstract.........................

\begin{abstract} 
A modern processor can dynamically set it's speed while it's active, and can make a transition to sleep state when required. When the processor is operating at a speed $s$, the energy consumed per unit time is given by a convex power function $P(s)$ having the property that $P(0) > 0$ and $P''(s) > 0$ for all values of $s$. Moreover, $C >  0$ units of energy is required to make a transition from the sleep state to the active state. The jobs are specified by their arrival time, deadline and the processing volume. 

We consider a scheduling problem, called {\it speed scaling with sleep state}, where each job has to be scheduled within their arrival time and deadline, and the goal is to minimize the total energy consumption required to process these jobs. Albers et. al. 
\cite{albers} proved the NP-hardness of this problem by reducing an instance of an NP-hard {\it partition problem} to an instance of this scheduling problem. The instance of this scheduling problem consists of the arrival time, the deadline and the processing volume for each of the jobs, in addition to $P$ and $C$. Since $P$ and $C$ depend on the instance of the {\it partition problem}, this proof of the NP-hardness of the {\it speed scaling with sleep state} problem doesn't remain valid when $P$ and $C$ are fixed. In this paper, we prove that the {\it speed scaling with sleep state} problem remains NP-hard for any fixed positive number $C$ and convex $P$ satisfying $P(0) > 0$ and $P''(s) > 0$ for all values of $s$. 

\end{abstract}

\section{Introduction}

A modern processor can dynamically set it's speed while it's active, and can make a transition to sleep state when required. When the processor is operating at a speed $s$, the energy consumed per unit time is given by a convex power function $P(s)$ having the property that $P(0) > 0$ and $P''(s) > 0$ for all values of $s$. Therefore, some energy is consumed  even if the processor is not scheduling any job in the active state. On the other hand, no energy is consumed when the processor is in the sleep state. However, $C >  0$ units of energy is required to make a transition from the sleep state to the active state and therefore it is not always fruitful to  go asleep when there is no work to be processed at some point of time. We assume that no energy is required to make a transition from the active state to the sleep state, as we can always include this energy requirement in the sleep to active state transition. A number of problems have been studied under this model, e.g., 
\cite{albers}, \cite{bampis}, \cite{baptiste}, \cite{baptiste_chrobak}, \cite{bansal_chan_lam}, 
\cite{bansal_chan_pruhs}, \cite{chan}, \cite{han}, \cite{irani}, \cite{irani_pruhs}, \cite{yao}.

The jobs are specified by their arrival time, deadline and the processing volume. We consider a scheduling problem where each job has to be scheduled within their arrival time and deadline, and the goal is to minimize the total energy consumption required to process these jobs. Albers et. al. \cite{albers} proved the
NP-hardness of this problem by reducing an instance of an NP-hard {\it partition problem} (defined below) to an instance of this scheduling problem. The instance of this scheduling problem consists of the arrival time, the deadline and the processing volume for each of the jobs, in addition to $P(s)$ and $C$ that depends on the problem instance of the partition problem. As a result, this proof of NP-hardness doesn't remain valid when we are given any fixed convex function $P(s)$ and a positive number $C$. In this paper, we prove that the problem remains NP-hard for any fixed positive number $C$ and convex function $P$ satisfying $P(0) > 0$ and $P''(s) > 0$ for all values of $s$. 

We would do the reduction from the following NP-hard {\it partition problem}: Given a finite set $A$ of $n$ positive integers $a_{1}, a_{2}, \ldots, a_{n}$, the problem is to decide whether there exists a subset 
$A' \subset A $  such that $\sum_{a_{i} \in A'} a_{i} = \sum_{a_{i} \notin A'} a_{i}$. It's assumed  
that $a_{max} \ge 2$; otherwise, the problem becomes trivial.

\section{The Reduction and it's Properties} 
\begin{figure}[t]
\centerline{{\resizebox*{4.8in}{1in}
{\includegraphics{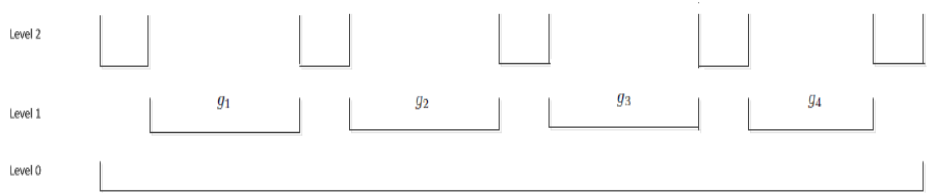}}
}}
\caption{An instance of $I_{S}$ with $n = 4$} 
\end{figure}
Let us start with a few definitions and notations. The {\it density} of an interval is defined as the total workload of the jobs that completely lie in an interval divided by the length of the interval. The critical speed $s*$ for $P(s)$ is defined as the minimum speed that minimizes $\frac{P(s)}{s}$. Note that the critical speed is not well-defined if $\frac{P(s)}{s}$ decreases monotonically. However, this is not a realistic case as this would mean we can schedule all jobs at an infinite speed to get the schedule that requires the minimum amount of energy consumption. Therefore, we assume that $\frac{P(s)}{s}$ decreases for $s < s*$ and attains the
minimum at $s = s*$. Under this assumption, the following property can be easily observed easily.

\begin{lemma}
\label{lemma1}
$P'(s) < \frac {P(s)}{s} < \frac {P(s*)}{s*} $ for $ s < s* $ and $ P'(s*) = \frac{P(s*)}{s*}$.
\end{lemma}
\begin{proof}
The derivative of the function $\frac{P(s)}{s}$ is $\frac {sP'(s)-P(s)} {s^2}$.
Since $\frac{P(s)}{s}$ decreases for $s < s*$, we have $ \frac {sP'(s)-P(s)} {s^2}\le 0$ for $s \le s*$ with equality only when $s = s*$. 
\qed
\end{proof}

Given the function $P(s)$, a non-negative number $C$, and an instance $I_{p} $ of the {\it partition problem}, i.e., the integers $a_1, a_2, \ldots, a_n$, let's define the parameters which would
be used to construct an instance $I_{S}$ of the scheduling problem. \\ \\
$R(s) = P(s) - \frac {P(s*)} {s*} s$.  \\ \\
$L_{i} = e - f a_{i} $ where $ e = C/R(d),  d = (1-\epsilon) s*,  0 < \epsilon <1/2$, and 
$f =  \frac {C} {P(0) a_{max}} $.\\ \\
$l_{i} = d L_{i} - \frac{a_{i}} {k} $ where $ k  = \frac {- R'(d)} {f R(d)} $. \\ \\
$ B = \frac{\sum\limits_{i=1}^n a_{i}} {2k} $. \\ \\
The structure of $I_S$ is the same as the one used by Albers et. al. \cite{albers}.
The job set in $I_S$ is partitioned into three levels. In level $0$, there is only one job having a processing volume equal to $B$. Level $1$ comprises of $n$ jobs; the $i$-th job has a processing volume $l_{i}$ and 
$L_{i}$ time units to process it. There are $n+1$ jobs in level $2$, with each job having a processing volume of $\delta s*$ and $\delta > 0$ time units to process it, thereby making the density of each of these jobs equal to $s*$.  

In the rest of this section, we establish a few lemmas that would be useful in our proof of 
NP-hardness.

\begin{lemma}
\label{lemma2}
$R(d) < \epsilon P(0)$.
\end{lemma}
\begin{proof}
Since $R(s) = P(s) - \frac{P(s*)}{s*}s$, we obtain $R'(s) = P'(s) -  \frac{P(s*)}{s*}$. We also note that
$R(0) = P(0) > 0 $ and $R(s*) = 0$. 
Furthermore, we obtain from Lemma \ref{lemma1} that 
$R'(s*) = P'(s*) - \frac{P(s*)}{s*} = 0$,  and  
$R'(s) = P'(s) - \frac{P(s*)}{s*} < 0$ for $s< s*$. 

Along with the properties established above, $ R''(s) = P''(s) > 0 $ implies that the following relationship holds.
\[R(d) < (1-\frac{d}{s*}) R(0) + \frac{d}{s*} R(s*) \Rightarrow  R(d) < \epsilon P(0).  \]
\qed
\end{proof}

\begin{lemma}
\label{lemma3}
$\frac{R(d)}{|R'(d)|} < \epsilon s*.$
\end{lemma}
\begin{proof}
It can be easily seen from Figure \ref{figure2} that $|R'(d)|>$ $\frac{R(d)}{s* - d}$ $=$ slope of line 
$ab$. Since $d = (1-\epsilon) s*$, it follows that $\frac{R(d)}{|R'(d)|} < \epsilon s*$.
\qed
\end{proof}

\begin{figure}[t]
\label{figure2}
\centerline{{\resizebox*{3in}{2in}
{\includegraphics{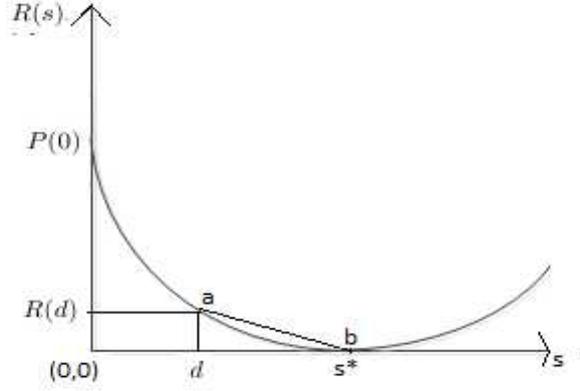}}
}}
\caption{Plot of $s$ vs $R(s)$} 
\end{figure}
We would now show that our choices of $l_{i}$ and $L_{i}$ satisfy the trivial constraints $l_{i}, L_{i} > 0$.   We would also that the density of all the intervals except those corresponding to level $2$ jobs are strictly less than $s*$.

\begin{lemma}
\label{lemma4}
$L_{i}, l_{i} > 0 $ for all $i$.
\end{lemma}
\begin{proof}
We first prove that $L_{i} > 0$ for all $i$. Note that
\[  L_{i} > 0 \Leftrightarrow e - f a_{i} > 0 \Leftrightarrow e > f a_{i}. \]
Since $a_{max} > a_{i} $, it suffices to show that $ e > f a_{max} $. As shown below, it can be easily seen using Lemma \ref{lemma2}, $C>0$ and $0 < \epsilon < \frac{1}{2}$.
\[ e > f a_{max} \Leftrightarrow \frac{C}{R(d)} > \frac{C}{P(0)} \Leftrightarrow R(d) < P(0). \] 
%  \[ L_{i} > 0  \forall i\]

\noindent In order to show that $l_{i} > 0 $, we observe that
\[ l_{i} > 0 \Leftrightarrow d L_{i} > \frac {a_{i}}{k} \Leftrightarrow d (e - f a_{i}) > \frac {a_{i}} {k} \Leftrightarrow  d > \frac {a_{i}} {k (e-fa_{i})} \Leftrightarrow d > \frac {1} { k (\frac{e}{a_{i}} - f)}. \]
Since $ a_{max} > a_{i} $, it suffices to show that $d > \frac{1}{k\left(\frac{e}{a_{max}} - f\right)}$.
We show it below using Lemma \ref{lemma2} and Lemma \ref{lemma3}.
\begin{eqnarray*}
\frac{1}{k\left(\frac{e}{a_{max}} - f\right)} & = & \frac{f R(d)} {|R'(d)| (\frac{C}{R(d) a_{max}} -f)}  \\
& = & \frac{R(d)} {|R'(d)| (\frac{C}{f R(d) a_{max}} - 1)}   \\
& = & \frac{R(d) R'(d)} {|R'(d)| (P(0) - R(d))} \\
& = & \frac{R(d)}{|R'(d)|} \cdot \frac{1}{(\frac{P(0)}{R(d)} - 1)} \\
& < & \epsilon s*  (\frac{1}{\frac{1}{\epsilon} - 1 }) \\
& = & \frac{\epsilon^2  s*}{1 - \epsilon} \\
& < & (1 - \epsilon ) s* = d.
\end{eqnarray*}
\noindent Note that the last inequality follows from our choice of $0 < \epsilon < \frac{1}{2}$. 
\qed
\end{proof}

\begin{figure}[t]
\centerline{{\resizebox*{3in}{2in}
{\includegraphics{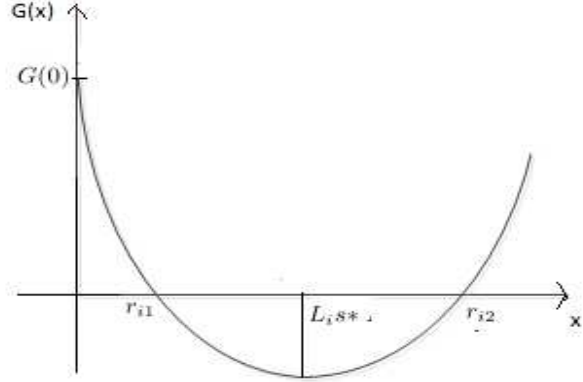}}
}}
\caption{Plot of $x$ vs $G(x)$} 
\end{figure}

\begin{lemma}
\label{lemma5}
The density of all the intervals except those corresponding to level $2$ jobs are strictly less than $s*$.
\end{lemma}
\begin{proof}
Let's first consider the intervals corresponding to level $1$ jobs. The density of such an interval is 
$\frac{ l_{i}} {L_{i}} $. Observe that 
\[  \frac{l_{i}} {L_{i}} < s* \Leftrightarrow \frac {d L_{i} - \frac {a_{i}}{k}} { L_{i}} > s*  \Leftrightarrow d < s* + \frac{a_{i}}{k L_{i}}. \] 
Since $d < s*$,  the density of an interval corresponding to a Level $1$ job is strictly less than $s*$. This also proves that the density of any interval corresponding to the union of a proper subset of the level $1$ and level $2$ jobs is less than $s*$, since such jobs are non-overlapping and the density of any interval corresponding a level $2$ job is exactly $s*$.   

\noindent Finally, we consider the interval that starts from the first arrival of the level $0$ job and lasts till it's deadline. The density of this interval is 
$\frac {(n+1) \delta s* + B + \sum_{i} l_{i} } {(n+1) \delta + \sum_{i} L_{i} } $. This quantity is less than 
$s*$ since
\[  \frac {B + \sum_{i} l_{i} } { \sum_{i} L_{i} } < s* \Leftrightarrow \frac { \sum_{i} a_{i} } {2k} + d \sum_{i} L_{i} - \sum_{i} \frac {a_{i}} {k} < s* \sum_{i} L_{i}  \Leftrightarrow d < s* + \frac { \sum_{i} a_{i}} { 2k \sum_{i} L_{i} }. \]
The last inequality is true since $d < s* $.
\qed
\end{proof}

Let us introduce the functions $F(x) = \frac{P(s*)}{s*} x + C $ and $H_{i}(x) = P(\frac{x}{L_{i}}) L_{i}$, and establish some of their properties.

\begin{lemma}
\label{lemma6}
$ F(x) $ and $ H_{i}(x) $ intersect at two different points for any $i$.
\end{lemma}
\begin{proof}
Consider $G(x) = H_{i}(x) - F(x)  =  P(\frac{x}{L_{i}}) L_{i} - \frac{P(s*)}{s*} x - C$. Note that
\begin{eqnarray*}
G(0) & = & P(0)L_{i} - C \\
& = & P(0)(e-fa_{i}) - C  \\
& = & P(0) (\frac{C}{R(d)} - \frac {C a_{i}}{P(0) a_{max}}) - C \\
& \ge & P(0) (\frac{C}{R(d)} - \frac{C}{P(0)}) - C \\
& = & \frac{C (P(0) - R(d))}{R(d)} - C \\
& > & C(\frac {P(0) (1-\epsilon)}{\epsilon P(0)}) - C \\
& = & C(\frac{1-2\epsilon}{\epsilon}) \\
& > & 0.
\end{eqnarray*}

\noindent The last-but-one inequality follows from Lemma \ref{lemma2}. The last inequality follows since 
$0 < \epsilon < 1/2$.

\noindent Next, we would show that $G(x)$ decreases for $x < L_{i} s*$, attains minimum at $x = L_{i}s*$  and then finally increases. Note that 
$G'(x) = L_{i} P'(\frac{x}{L_{i}}) \frac{1}{L_{i}} - \frac{P(s*)}{s*} = P'(\frac{x}{L_{i}}) -\frac{P(s*)}{s*}$.
The following inequalities follow easily from Lemma \ref{lemma1} and from the fact that $P''(s) > 0$.
\[ G'(x) \le 0 \Leftrightarrow P'(\frac{x}{L_{i}}) \le  \frac{P(s*)}{s*}  \Leftrightarrow P'(\frac{x}{L_{i}}) \le P'(s*) \Leftrightarrow x \le L_{i}s* .\]
By Lemma \ref{lemma1}, the inequalities above would be strict since $x < L_{i} s*$. \\ \\
We complete the proof by showing that $ G(L_{i} s*) = P(s*) L_{i} - P(s*) L_{i} - C = -C < 0 $. Since $G$ is a strictly convex function (note that $ G''(x) = \frac{1}{L_{i}} P''(\frac{x}{L_{i}}) $), it 
would eventually intersects the $x$-axis at some point $ x > L_{i} s* $.  
\qed
\end{proof}

\begin{lemma}
\label{lemma7}
$H'_{i}(l_{i}+\frac{a_{i}}{k}) = \frac {H_{i}(l_{i} + \frac{a_{i}}{k}) - F(l_{i})} {a_{i}/k} = P'(d)$.
\end{lemma}
\begin{proof}
It can be easily seen that $H'_{i}(l_{i}+\frac{a_{i}}{k}) = H'_{i}(dL_{i}) = P'(\frac{dL_{i}}{L_{i}}) = P'(d)$.
The following calculation also shows that $ \frac {H_{i}(l_{i} + \frac{a_{i}}{k}) - F(l_{i})} {a_{i}/k} = P'(d)$.
\begin{eqnarray*}
\frac {H_{i}(l_{i} + \frac{a_{i}}{k}) - F(l_{i})} {a_{i}/k} & = &\frac{k}{a_{i}} [H_{i}(dL_{i}) - F(l_{i})] \\
& = &\frac {k}{a_{i}} [P(d) L_{i} - \frac{P(s*)}{s*} l_{i} - C] \\
& = & \frac{k}{a_{i}} [P(d) L_{i} - \frac{P(s*)}{s*} d L_{i} +\frac{P(s*)}{s*} \frac{a_{i}}{k} - C] \\
& = & \frac{k}{a_{i}} R(d) L_{i} + \frac{P(s*)}{s*} - \frac{Ck}{a_{i}}	\\
& = & \frac{k}{a_{i}} R(d)e - \frac{k}{a_{i}} R(d) f a_{i} +\frac {P(s*)}{s*} - \frac{Ck}{a_{i}} \\
& = & \frac{k}{a_{i}} C +R'(d) + \frac{P(s*)}{s*} - \frac{Ck}{a_{i}} \\
& = & P'(d) - \frac{P(s*)}{s*} + \frac{P(s*)}{s*}  \\
& = & P'(d).
\end{eqnarray*}
\qed
\end{proof}

Let $r_{i1}$ and $r_{i2}$ be the two roots of the equation $ F(x) = H_{i}(x) $ such that $ r_{i1} < r_{i2}$.
We establish the following two lemmas.

\begin{lemma}
\label{lemma8}
$0 < l_{i} < r_{i1}$, where $r_{i1} $ is the first intersection of $H_{i}(x)$ and $F(x)$.
\end{lemma}
\begin{proof}
Since $ H_{i}(x)$ is strictly convex at every $x$, we obtain
\begin{eqnarray*}
H'_{i}(l_{i} + \frac{a_{i}}{k}) > \frac {H_{i}(l_{i}+\frac{a_{i}}{k}) - H_{i}(l_{i})} {a_{i}/k} 
& \Rightarrow & \frac {H_{i}(l_{i} + \frac{a_{i}}{k}) - F(l_{i})} {a_{i}/k} > \frac {H_{i}(l_{i}+\frac{a_{i}}{k}) - H_{i}(l_{i})} {a_{i}/k} \\
& \Rightarrow & -F(l_{i}) > - H_{i}(l_{i}) \\
& \Rightarrow & G(l_{i}) >0. 
\end{eqnarray*}
The lemma follows since $l_{i} < L_{i}s*$.
\qed
\end{proof}

\begin{lemma}
\label{lemma9}
$r_{i1} < (l_{i} + \frac{a_{i}}{k})< r_{i2}$.
\end{lemma}
\begin{proof}
Assume for the sake of contradiction that $ l_{i} + \frac{a_{i}}{k} \le r_{i1}$. It can be seen from 
Figure \ref{figure4} that it implies $\frac {H_{i}(l_{i} + \frac{a_{i}}{k}) - F(l_{i})} {a_{i}/k}  \ge \frac{P(s*)}{s*} $  $\Rightarrow P'(d) \ge \frac{P(s*)}{s*}$. However, this leads to a contradiction since
$d < s*$, and $P'(d) < \frac{P(s*)}{s*} $ by Lemma \ref{lemma1}. 
On the other hand, it's easy to see that $ l_{i} + \frac{a_{i}}{k} = d L_{i} < L_{i}s* < r_{i2} $ . 
\qed
\end{proof}

\begin{figure}[t]
\centerline{{\resizebox*{3in}{2in}
{\includegraphics{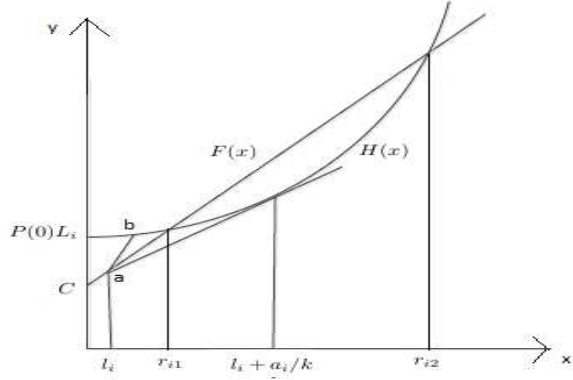}}
}}
\caption{Plot of $x$ vs $F(x)$ and $H(x)$}
\label{figure4} 
\end{figure}

\section{Proof of NP-hardness}  
In this section, we would complete the proof of the NP-hardness of the {\it speed scaling with power down} problem. We would be using the following result by Irani et al. \cite{irani} along with the results derived in the previous section.
\begin{lemma}
\label{lemma10}
\cite{irani} There exists an optimal solution of the ``speed scaling with a sleep state" problem that satisfies the following properties: 
\begin{itemize} 
\item A job $j$ must be scheduled at a constant speed $ s_{j}. $ \\
\item Suppose that the arrival time and the deadline of a job $j$ is $r_{j}$ and $d_{j}$, respectively. If 
another job $k$ is scheduled in the interval $[r_{j}, d_{j}]$, then $s_{k} \ge s_{j} $. \\
\item The jobs in the intervals having density at least $s*$ are scheduled according to the 
YDS algorithm \cite{yao}. The YDS algorithm is an iterative algorithm. In each iteration, an interval with the maximum density is identified and an {\it earliest-deadline-first} policy is used to construct a schedule for the jobs that lie completely in that interval. After an iteration, the YDS algorithm removes the jobs that lie completely in the maximum density interval corresponding to that iteration, and updates the arrival time
and deadline of any job that overlaps with
that interval.
\end{itemize}
\end{lemma}
%We will start by lemma given by Irani.et.al [3].  lemma x: YDS in dense intervals. lemma y : speed constraint.\\
% \textbf{Lemma }

\begin{theorem}
\label{theorem1} An instance $ I_{p} $ of the partition problem admits a partition if and only if there exists a a feasible schedule for $ I_{s} $ with total energy consumption of at most  $(n+1) \delta P(s*) + \sum \limits_{i=1}^n F(l_{i}) + B P'(d)$. 
\end{theorem}
\begin{proof}
$(\Rightarrow) $ Let's first assume that $ I_{p} $ admits a partition and construct a feasible schedule of energy at most  $(n+1) \epsilon P(s*) + \sum \limits_{i=1}^n F(l_{i}) + B R_{min}$. We start with some notations.

Let $A'$ be the set of i corresponding to the  solution of Partition problem, i.e., $ \sum \limits_{i \in A'} a_{i} = \frac{\sum \limits_{i=1}^n a_{i}}{2}$. Let $ b_{i} $ denote the portion of the workload of Level 0 job scheduled in gap $ g_{i}$. It can be seen that $ \sum \limits_{i=1}^n b_{i} = B $. We set the $b_{i}$'s as follows: 
%\begin{equation}
\begin{center}
$ b_{i} = \left\lbrace
\begin{array}{r l}
a_{i}/k, & \mbox{if }i\in A' \\
0, & \mbox{otherwise}.
\end{array}
\right. $
\end{center}
%\label{eqn:abs}
%\end{equation}  
Our schedule executes any Level $2$ jobs with speed $s*$ between it's release time and deadline. This is feasible since the density of any such job is equal to $s*$. Therefore, a total workload of $ l_{i} + a_{i}/k = d L_{i}$ has to be scheduled in each gap $g_{i}$ corresponding to an $i \in A'$. We schedule both the jobs in gap $g_i$ with speed $d$. In the rest of the gaps, the Level $1$ jobs are scheduled at speed $s*$. The processor transitions to the sleep state at the completion of the job in such gaps, and wakes up at the release time of a Level $2$ job. Since the density of any interval corresponding to a Level $1$ job is less than $s*$, we get a feasible schedule. 

Let us calculate the total energy consumed by the jobs at every level. First of all, the total energy consumed by the Level 2 jobs is $(n+1) \delta P(s*) $. In the gaps corresponding to $ i \in A'$, we note that the jobs are proceeded at a speed $ d = (\frac { l_{i} + a_{i}/k} {L_{i}}) $ for $L_{i}$ units of time. The energy consumption in such a gap equals $ P(\frac { l_{i} + a_{i}/k} {L_{i}}) L_{i}$, which is the same as $ H_{i}(l_{i}+ \frac {a_{i}} {k})$. In a gap corresponding to $ i\notin A'$, a total $l_{i}$ units of workload  are scheduled at speed $ s* $  and then the processor transitions to sleep state. Therefore, the energy consumed is given by  
$P(s*)\frac {l_{i}}{s*} + C $, which is the same as $ F(l_{i})$. From lemma \ref{lemma7}, $ H_{i}(l_{i}+ \frac {a_{i}} {k})$ can be written as $ F(l_{i}) + P'(d) \frac {a_{i}}{k} $.   \\ \\
Let $ E_{0,1} $ denote the total Energy consumed by the Level $0$ and Level $1$ jobs. We obtain the following.
\begin {eqnarray*}
E_{0,1} & = & \sum \limits_{i \in A'} H_{i}(l_{i} + \frac{a_{i}} {k})   + \sum \limits_{i \notin A'} F(l_{i}) \\
&  = & \sum \limits_{i \in A'} ( F(l_{i}) + P'(d) \frac{a_{i}} {k}) + \sum \limits_{i \notin A'} F_{i}(l_{i})\\
& = & \sum \limits_{i = 1}^n F(l_{i}) + P'(d) \sum \limits_{i \in A'} \frac {a_{i}}{k} \\
& =  & \sum \limits_{i = 1}^n F(l_{i}) + P'(d) B.
\end{eqnarray*}  \\
The last equality follows since $ \sum \limits_{i = 1}^n b_{i} = \sum \limits_{i \in A'} a_{i}/k = (\sum \limits_{i =1}^n a_{i}) /2k = B $.  Hence, we get a feasible schedule whose total energy consumption is  $ (n+1) \delta P(s*) + \sum \limits_{i = 1}^n F(l_{i}) + P'(d) B $.

\begin{figure}[t]
\centerline{{\resizebox*{2.6in}{2in}
{\includegraphics{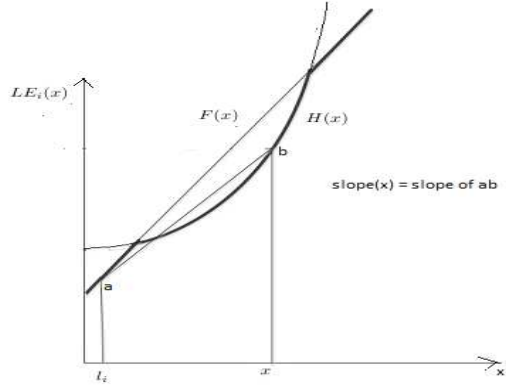}}
}}
\caption{Plot of $x$ vs the lower envelope function $LE_{i}(x)$}
\label{figure5} 
\end{figure}

($ \Leftarrow $) In the reverse direction of the proof, we assume that $I_{p}$ doesn't admits a partition and  show that the energy consumption in any feasible schedule is strictly greater than  $(n+1) \delta P(s*) + \sum \limits_{i=1}^n F(l_{i}) + B P'(d)$. 

Let $ LE_{i}(x) = min \{F(x),H_{i}(x) \}$ denote the lower envelope of the functions $ F(x) $ and $ H_{i}(x)$, represented by solid curve in Figure \ref{figure5}.  Let slope(x) denote the slope of the line joining $ (l_{i}, LE_{i}(l_{i}))$ and $ (x, LE_{i}(x)) $. For $x \ge l_{i} $, $ LE_{i}(x) $ can be written as 
$L E_{i}(x) = LE_{i}(l_{i}) + ( LE_{i}(x) - LE_{i}(l_{i}))  =  F(l_{i}) + (slope(x))* (x - l_{i})$.
We note that the slope(x) is minimum at $ x = l_{i} + \frac{a_{i}} {k} $  and the minimum value is $ H'(l_{i} + \frac {a_{i}} {k}) = P'(d)$ (by Lemma \ref{lemma7}) which is  independent of $i$. 

%\begin{flushleft}
%\includegraphics[scale = .9]{aa}
%\end{flushleft}
Consider an Optimal schedule $ S $ satisfying the properties of lemma \ref{lemma10} and let $ b_{1},b_{2},....,b_{n} $ units of workload of Level 0 job be scheduled in the gaps $ g_{1},g_{2},.....,g_{n}$, respectively. 
Let $ A' = \{ i | r_{i1} \le l_{i} + b_{i} \le r_{i2} \} $. \\ \\
\textbf{Case 1.} $ b_{i} = \frac {a_{i}}{k} $ for some $ i \in A' $  \\ \\
 Since the workload $ l_{i} + b_{i} $  is greater than $r_{i1} $ and less than $r_{i2} $, it is beneficial to schedule it at the speed $ (l_{i} + b_{i})/L_{i} $ (rather than to schedule it with the speed $s*$) and then transition to sleep state. From Lemma \ref{lemma10}, it follows that the ratio $ (l_{i} + b_{i})/L_{i} $ must be the same for all $ i \in A' $ in the schedule $ S$. 
Take $ i \in A' $ corresponding to $ b_{i} = \frac {a_{i}} {k} $. We show below that $ b_{j} $ must also be equal to $ a_{j}/k $  for all $  j \in A' $ in an optimal schedule. 
\begin{eqnarray*} 
\frac{(b_{i} + l_{i})}{L_{i}} = \frac  {(b_{j} + l_{j})}{L_{j}} 
& \Rightarrow & \frac {a_{i}/k + dL_{i} - a_{i}/k} { L_{i}} = \frac {b_{j}+ dL_{j} - a_{j}/k} { L_{j}}  \\
& \Rightarrow & d = d + \frac { b_{j} - a_{j}/k } {L_{j}} \\
& \Rightarrow & b_{j} = a_{j}/k.
\end{eqnarray*} 

Lemma \ref{lemma10} says that all the intervals having density greater than or equal to s* must be scheduled according to YDS in the schedule $S$. Also, Lemma \ref{lemma1} tells that all the  intervals except those corresponding to Level 2 jobs are having density less than s*. Therefore, in the schedule $S$,  all Level 2 jobs must be scheduled at s*. Thus, the total energy consumed by the Level 2 jobs is $ (n+1) \delta P(s*) $. 

Let us again denote the total energy required by the Level 0 and level 1 jobs as $E_{0,1}$.  In a gap corresponding to $i \notin A'$, it is optimal to schedule the job at speed $s*$ and then transition to sleep state than scheduling at the speed  $ (l_{i} + b_{i})/L_{i} $. When this is feasible, the energy consumption in the gap would be given by $ LE(l_{i} +b_{i}) $. When it's not (i.e., if $ (l_{i} + b_{i})/L_{i} > s*$), the energy consumption in the gap would be greater than $ LE_{i}(l_{i} +b_{i}) $. Therefore, we obtain the following lower bound on $E_{0,1}$.
\begin{align*}
 E_{0,1} & \ge  \sum \limits_{i \in A'} LE_{i}(l_{i} + \frac{a_{i}} {k})   + \sum \limits_{i \notin A'} LE_{i}(l_{i} + b_{i}) \\
& =  \sum \limits_{i \in A'} (F(l_{i}) + P'(d) \frac{a_{i}} {k})   + \sum \limits_{i \notin A'} (F(l_{i}) + \frac{P(s*)}{s*} b_{i}) \\
& = \sum \limits_{i = 1}^n F(l_{i}) + P'(d) \sum \limits_{i \in A'} a_{i}/k + \frac {P(s*)}{s*} (B - \sum \limits_{i \in A'}b_{i}).  
\end{align*}

If $\sum \limits_{i \in A'}b_{i} = B$, it implies that $\sum \limits_{i \in A'} \frac{a_{i}}{k}  = \frac {\sum \limits_{i = 1}^n a_{i}} {2k}$, which contradicts our assumption that a solution of the partition problem does not exist. Therefore, $\sum \limits_{i \in A'}b_{i} < B$, which implies that  
$\sum \limits_{i \in A'} a_{i}/k < B$. The following calculation completes the proof of Case $1$.
\begin{align*}
 E_{0,1} &  \ge  \sum \limits_{i = 1}^n F(l_{i}) + P'(d) \sum \limits_{i \in A'} a_{i}/k +  \frac  {P(s*)}{s*} ( B - \sum \limits_{i \in A'} a_{i}/k) \\
& =  \sum \limits_{i = 1}^n F(l_{i}) + B \frac {P(s*)}{s*} - (\sum \limits_{i \in A'} a_{i}/k) (  \frac {P(s*)}{s*} - P'(d)) \\
& >  \sum \limits_{i = 1}^n F(l_{i}) + B \frac {P(s*)}{s*} - B( \frac {P(s*)}{s*} - P'(d)) \\
& =  \sum \limits_{i = 1}^n F(l_{i}) + B P'(d)
  \end{align*}
\\ \\
\textbf{Case 2.} $b_{i} \neq a_{i}/k$  for all $i \in A'$  \\ \\

In this case, the following calculation completes the proof.
\begin{align*}
E_{0,1} & = \sum \limits_{i \in A'} (F(l_{i}) + slope(l_{i}+b_{i}) * b_{i}) + \sum \limits_{i \notin A'} (F(l_{i}) + slope(l_{i}+b_{i}) * b_{i}) \\
 & > \sum \limits_{i \in A'} F(l_{i}) + P'(d) \sum \limits_{i \in A'} b_{i}  + \sum \limits_{i \notin A'} F(l_{i}) + P'(d) \sum \limits_{i \notin A'} b_{i} \\
& =  \sum \limits_{i = 1}^n F(l_{i}) + B P'(d).
\end{align*} 
\qed
\end{proof}

\end{document}